\documentclass[a4paper]{article}
\usepackage{amssymb,amsmath,amsthm}
\usepackage{fullpage}
\usepackage{hyperref}
\usepackage{graphicx}
\usepackage{xcolor}
\usepackage{float}
\usepackage{stackrel}
\newcommand{\ie}{\emph{i.e.}}
\newcommand{\eg}{\emph{e.g.}}
\newcommand{\cf}{\emph{cf.}}
\newcommand{\Real}{\mathbb{R}}
\newcommand{\Nat}{\mathbb{N}}

\newcommand{\supp}{\mathop{\mathrm{supp}}\nolimits}
\newcommand{\dist}{\mathop{\mathrm{dist}}\nolimits}
\newcommand{\dom}{\mathop{\mathrm{dom}}\nolimits}

\newcommand{\Cut}{\mathrm{Cut}}
\newcommand{\der}{\mathrm{d}}

\newcommand{\sii}{L^2}
\newtheorem{Theorem}{Theorem}
\newtheorem{Proposition}{Proposition}

\newtheorem{Corollary}{Corollary}
\theoremstyle{definition}
\newtheorem{Remark}{Remark}
\usepackage[normalem]{ulem}
\definecolor{DarkGreen}{rgb}{0,0.5,0.1} 
\newcommand{\txtD}{\textcolor{DarkGreen}}

\newcommand\soutD{\bgroup\markoverwith
{\textcolor{DarkGreen}{\rule[.5ex]{2pt}{1pt}}}\ULon}
\newcommand{\Hm}[1]{\leavevmode{\marginpar{\tiny%
$\hbox to 0mm{\hspace*{-0.5mm}$\leftarrow$\hss}%
\vcenter{\vrule depth 0.1mm height 0.1mm width \the\marginparwidth}%
\hbox to
0mm{\hss$\rightarrow$\hspace*{-0.5mm}}$\\\relax\raggedright #1}}}

\begin{document}
%
\title{\Large\textbf{%
Soft quantum waveguides with an explicit cut-locus}}%
\author{Sylwia Kondej\,$^a$, \ David Krej\v{c}i\v{r}{\'\i}k\,$^b$ \
and \ Jan K\v{r}\'i\v{z}\,$^c$}	
\date{\small
\begin{quote}
\emph{
\begin{itemize}
\item[$a)$]
Institute of Physics, University of Zielona G\'ora, ul.\ Szafrana
4a, 65246 Zielona G\'ora, Poland;
s.kondej@if.uz.zgora.pl%
\item[$b)$]
Department of Mathematics, Faculty of Nuclear Sciences and
Physical Engineering, Czech Technical University in Prague,
Trojanova 13, 12000 Prague 2, Czechia;
david.krejcirik@fjfi.cvut.cz.%
\item[$c)$]
Department of Physics, Faculty of Science, University of Hradec Kr\'alov\'e,
Rokitansk\'eho 62, 500 03 Hradec Kr\'alov\'e, Czechia;
jan.kriz@uhk.cz
\end{itemize}
}
\end{quote}
21 July 2020}
\maketitle
\begin{abstract}
\noindent
We consider two-dimensional Schr\"odinger operators
with an attractive potential in the form of a channel of a fixed
profile built along an unbounded curve composed of
a circular arc and two straight semi-lines.
Using a test-function argument with help of parallel coordinates
outside the cut-locus of the curve,
we establish the existence of discrete eigenvalues.
This is a special variant of a recent result of Exner~\cite{Exner_2020}
in a non-smooth case and via a different technique which does not require non-positive constraining potentials.
%
%
\end{abstract}
%

\section{Introduction}
%
Nanotechnology make the conceptual model of
a quantum particle propagating in the vicinity of
an unbounded curve~$\Gamma$ in~$\Real^2$ a realistic system.
Mathematically, the model is described by the Schr\"odinger operator
\begin{equation}\label{Hamiltonian}
  H := -\Delta + V
  \qquad\mbox{in}\qquad
  \sii(\Real^2)
  \,,
\end{equation}
where $V:\Real^2 \to \Real$ is a potential modelling
a force which constrains the particle to the tubular neighbourhood
\begin{equation}\label{tube}
  \Omega_a := \{x\in\Real^2 : \ \dist(x,\Gamma)<a\}
  \,,
\end{equation}
where~$a$ is a positive constant such that~$\Omega_a$
does not overlap itself.

The \emph{hard-wall} idealisation
\begin{equation}\label{hard}
  V_\mathrm{hard}(x) :=
  \begin{cases}
     0 & \mbox{if} \quad x \in \Omega_a \,, \\
     \infty & \mbox{otherwise} \,, \\
  \end{cases}
\end{equation}
is certainly best understood (mathematically, $H_\mathrm{hard}$~is realised as
the Laplacian in $\sii(\Omega_a)$, subject to Dirichlet boundary conditions).
In 1989 Exner and \v{S}eba~\cite{ES} applied the Birman--Schwinger principle
and established the existence of discrete eigenvalues of~$H_\mathrm{hard}$
provided that~$\Gamma$ is not straight,
it is straight asymptotically in a suitable sense
and~$a$ \emph{is sufficiently small}.
Applying variational tools instead, Goldstone and Jaffe~\cite{GJ}
removed the last, smallness hypothesis,
making the existence of curvature-induced discrete spectra
in quantum waveguides with hard-wall boundaries a universal fact.
We refer to~\cite{KKriz} for a proof under minimal hypotheses
and further references.

The defect of the hard-wall model~\eqref{hard} is that it completely
disregards tunnelling effects.
As an alternative, in 2001 Exner and Ichinose~\cite{EI}
came with the model of \emph{leaky wires}:
\begin{equation}\label{leaky}
  V_\mathrm{leaky}(x) := \alpha \, \delta(x-\Gamma)
  \,, \qquad \alpha < 0 \,,
\end{equation}
where~$\delta$ is the Dirac delta function
(mathematically, $H_\mathrm{leaky}$ is introduced as
the Laplacian in $\sii(\Real^2)$, subject to customary jump conditions
along the interface~$\Gamma$).
Again by the Birman--Schwinger principle,
it was demonstrated in~\cite{EI} that the discrete spectrum of~$H_\mathrm{leaky}$
exists under geometric hypotheses which are however less satisfactory than
in the hard-wall case
(see~\cite{Exner_2008} for a review and further references).

The leaky model~\eqref{leaky} is another extreme for
the constraining potential is of zero-range.
As an intermediate situation, just recently Exner~\cite{Exner_2020}
introduced the model of \emph{soft waveguides}:
\begin{equation}\label{soft}
  V_\mathrm{soft}(x)
  \begin{cases}
     \mbox{is bounded and non-positive}
     & \mbox{if} \quad x \in \Omega_a \,, \\
     =0 & \mbox{otherwise} \,. \\
  \end{cases}
\end{equation}
The situation of special interest is that of quantum square well,
\ie~when $V_\mathrm{soft}$
equals a negative constant inside~$\Omega_a$.
Using the Birman--Schwinger principle,
sufficient conditions guaranteeing
that the discrete spectrum of~$H_\mathrm{soft}$ is not empty
were derived in~\cite{Exner_2008}.
In particular, the discrete eigenvalues exist whenever~$V_\mathrm{soft}$
is ``deep and narrow enough''.

It is worth mentioning that the concept of soft waveguides is implicitly included
in the work~\cite{WT} (see also~\cite{HLT})  preceding~\cite{Exner_2020}.
The obtained effective Hamiltonian of~\cite{WT}
in the limit of \emph{thin} soft waveguides can be used
to study the existence of the discrete spectrum in this asymptotic regime.

The purpose of the present note is to make a small observation
that there is a very specific class of soft waveguides for
which the existence of discrete spectra can be proved in the full generality
and directly via customary variational tools.
Moreover, we are able to consider the leaky wires~\eqref{leaky} at the same time,
so we establish new results for this model too.
The hard-wall waveguides~\eqref{hard} could be also treated simultaneously,
but our technique does not bring anything new in this case.
As a matter of fact, our modus operandi is based on developing the method
of parallel coordinates based on~$\Gamma$ involving the \emph{cut-locus} of~$\Gamma$.
The latter has an empty intersection with~$\Omega_a$,
so the approach is actually identical to~\cite{KKriz} for the hard-wall waveguides.
It is the presence of a non-trivial cut-locus of~$\Gamma$ in the whole space~$\Real^2$
which makes the present approach unprecedented.
Unfortunately, however, our argument works only because of
an explicit knowledge of the structure of the cut-locus of our specific curve~$\Gamma$.
On the other hand,  it is worth mentioning
that the method of this paper does not require the assumption of non-positivity of the constraining potential~(\ref{soft}).
This restriction made in~\cite{Exner_2020}
is substituted here
by the more general condition that the corresponding one-dimensional potential formed by taking cross-section produces at least one negative eigenvalue.
We hope that this new idea will be appreciated by the community
interested in quantum waveguides and further developed to other problems subsequently.

The structure of the paper is as follows.
In Section~\ref{Sec.smooth} we start with a general presentation
of the cut-locus idea for smooth curves.
Our non-smooth model is introduced in Section~\ref{Sec.model}.
The proof of the existence of discrete eigenvalue for the latter
is performed in Section~\ref{Sec.proof}.

\section{Parallel coordinates in the smooth case}\label{Sec.smooth}
%
Let us first explain the main idea in the usual case
of \emph{smooth} curves~$\Gamma$.
More specifically, in agreement with~\cite[Ass.~(a)]{Exner_2008},
let $\Gamma:\Real\to\Real^2$ be a $C^2$-smooth curve.
Without loss of generality, we suppose that~$\Gamma$ is parameterised
by its arc-length, \ie, $|\dot\Gamma(s)|=1$ for every $s \in \Real$.
We introduce $N(s) := (-\dot\Gamma^2(s),\dot\Gamma^1(s))$,
the unit vector normal to~$\Gamma$ at~$s$
and oriented in such a way that the Frenet frame $(\dot\Gamma,N)$
has the same orientation as~$\Real^2$.
The (signed) \emph{curvature}~$\kappa$ of~$\Gamma$ is defined
by the Frenet equation $\ddot\Gamma=\kappa N$.
More specifically,
$
  \kappa(s)
  = \dot\Gamma^1(s) \ddot\Gamma^2(s) - \dot\Gamma^2(s) \ddot\Gamma^1(s)
  = - \gamma(s)
$
for every $s \in \Real$,
where~$\gamma$ is the signed curvature of~\cite{Exner_2008}.
In our convention, the curvature~$\kappa$ has a positive sign
if the curve~$\Gamma$ is turning toward the normal~$N$.

If $V : \Real^2 \to \Real$ is an essentially bounded function
(as is the case of the soft realisation~\eqref{soft}),
then the quantum Hamiltonian~\eqref{Hamiltonian}
can be introduced as an ordinary operator sum of the self-adjoint Laplacian
with domain $H^2(\Real^2)$ and a maximal operator of multiplication
generated by~$V$. The associated closed form reads
\begin{equation}\label{form}
  h[u] := \int_{\Real^2} |\nabla u|^2
  + \int_{\Real^2} V |u|^2
  \,, \qquad
  \dom h := H^1(\Real^2)
  \,.
\end{equation}
If~$V$ is the distribution of the leaky type~\eqref{leaky},
the simplest is to start with the form~\eqref{form},
where the second term should be interpreted as
$\alpha \int_{\Gamma} |u|^2$.
Again, it is a well-defined and closed form
under our standing hypothesis that there exists
a positive number~$a$ such that the tubular neighbourhood~\eqref{tube}
does not overlap itself.
In either case, $H$~can be \emph{defined} as the self-adjoint
operator associated with~$h$ (with the properly interpreted second integral)
via the representation theorem \cite[Thm.~VI.2.1]{Kato}.

Let us consider the \emph{normal exponential map}	
$\Phi:\Real^2 \to \Real^2$ by setting
\begin{equation}
  \Phi(s,t) := \Gamma(s) + N(s) \, t
  \,,
\end{equation}
which gives rise to \emph{parallel} (or \emph{Fermi})
``coordinates'' $(s,t)$ based on~$\Gamma$.
The crucial requirement that the tubular neighbourhood~\eqref{tube}
does not overlap itself is equivalent to the fact that
the restricted map $\Phi:\Real\times(-a,a) \to \Omega_0$
is a diffeomorphism.
Since the Jacobian of~$\Phi$ is given by
\begin{equation}\label{Jacobian}
  f(s,t) := 1-\kappa(s) \, t
  \,,
\end{equation}
it is clear that a necessary condition to ensure
the property is that~$\kappa$ is bounded.
Then any~$a$ satisfying the inequality
\begin{equation}\label{Ass}
  \fbox{$
  a \, \|\kappa\|_{L^\infty(\Real)} < 1
  $}
\end{equation}
ensures that the map~$\Phi$ induces a \emph{local} diffeomorphism.
To ensure that it is a \emph{global} diffeomorphism,
one usually assumes \emph{ad hoc} that
\begin{equation}\label{Ass.global}
  \fbox{
  $\Phi \upharpoonright \Real\times(-a,a)$ is injective.
  }
\end{equation}
Within~$\Omega_a$, one observes that
$s \mapsto \Phi(s,t)$ is a curve parallel to~$\Gamma$
at the distance~$|t|$ for any fixed $t \in (-a,a)$,
while $t \mapsto \Phi(s,t)$ is a straight line
(\ie\ geodesic in~$\Real^2$)
orthogonal to $\Gamma(s)$ for any fixed $s \in \Real$.

The crucial assumption of~\cite[Ass.~(e)]{Exner_2008}
in the case of soft waveguides
is that the profile of the constraining potential~$V$ does not vary along~$\Gamma$,
\ie,
\begin{equation}\label{Ass.vary}
  \fbox{
  \mbox{$W(t) := (V \circ \Phi)(s,t)$ is independent of~$s$}
  \qquad \& \qquad
  $\supp W \subset [-a,a]$
  .
  }
\end{equation}
This is certainly the case of leaky wires~\eqref{leaky} too,
because~$\delta$ is zero range
and~$\alpha$ is assumed to be a constant.

From now on, let us assume that~$\kappa$ is bounded
and that there exists a positive~$a$ such that~\eqref{Ass}
and~\eqref{Ass.global} hold.
Define the \emph{cut-radius} maps $c_\pm:\Real \to (0,\infty)$
by the property that the segment $t \mapsto \Phi(s,t)$
for positive (respectively, negative) $t$
minimises the distance from~$\Gamma$
if, and only if, $t \in [0,c_+(s)]$ (respectively, $t \in [-c_-(s),0]$).	
The cut-radius maps are known to be continuous.
The \emph{cut-locus}
\begin{equation}
  \Cut(\Gamma) :=
  \{\Phi(s,c_+(s)) : s \in \Real\} \cup
  \{\Phi(s,-c_-(s)) : s \in \Real\}
\end{equation}
is a closed subset of~$\Real^2$ of measure zero
(see, \eg, \cite[Chap.~\txtD{III}]{Chavel}).
The map~$\Phi$, when restricted to the open set
\begin{equation}
  U := \{ (s,t) \in  \Real^2
  : -c_-(s) < t < c_+(s)
  \}
\end{equation}
is a diffeomorphism onto $\Phi(U) = \Real^2 \setminus  \Cut(\Gamma)$.
Obviously, one has the inclusion
$
  \Cut(\Gamma) \supset \{\Phi(s,t) : f(s,t)=0\}
  =: \Cut_0(\Gamma)
$.

Outside the cut-locus, we have the usual coordinates of curved
quantum waveguides. More specifically, we define the unitary map
$
  \mathcal{U} : \sii(\Real^2) \to \sii(U)
$
by setting $\mathcal{U} u := u \circ \Phi$.
Then $\mathcal{U} H \mathcal{U}^{-1}$ is the operator associated
with the quadratic form $Q[\psi] := h[\mathcal{U}^{-1}\psi]$,
$\dom Q := \mathcal{U} \dom h \subset H^1(U)$.
An explicit computation yields
\begin{equation}\label{form.transformed}
  Q[\psi]
  = \int_U \frac{|\partial_s\psi|^2}{1-\kappa(s) t} \, \der s \, \der t
  + \int_U |\partial_t\psi|^2 (1-\kappa(s) t) \, \der s \, \der t
  + \int_U W(t) \, |\psi|^2 (1-\kappa(s) t) \, \der s \, \der t
  \,,
\end{equation}
where $W(t) := (V \circ \Phi)(s,t)$
Hereafter the last integral should be interpreted as
$\alpha \int_\Real |\psi(s,0)|^2 \, \der s$
in the case of leaky wires.

The problem is that the form domain $\dom Q$
is not easy to identify because of
the boundary conditions on~$\partial U$.
An objective of this paper is to show that
there exists a special class of curves for which this is feasible
because of an explicit knowledge of the cut-locus.
Consequently, the usual variational argument
for quantum waveguides apply.

\section{Special piecewise smooth cases}\label{Sec.model}
%
Given arbitrary numbers $R>0$ and $\theta \in [0,\pi]$,
a parameterisation of the special class of curves
that we address in this paper is given by:
\begin{equation}\label{curve}
  \Gamma(s) :=
  \begin{cases}
    \left(
    (s+\frac{\theta}{2} R) \cos\frac{\theta}{2} - R \sin\frac{\theta}{2},
    -(s+\frac{\theta}{2} R) \sin\frac{\theta}{2}
    + R (1-\cos\frac{\theta}{2})
    \right)
    & \mbox{if} \quad s \leq -\frac{\theta}{2} R
    \,,
    \\
    \left(
    R \sin \frac{s}{R},
    R (1-\cos \frac{s}{R})
    \right)
    & \mbox{if} \quad -\frac{\theta}{2} R < s < \frac{\theta}{2} R
    \,,
    \\
    \left(
    (s-\frac{\theta}{2} R) \cos\frac{\theta}{2} + R \sin\frac{\theta}{2},
    (s-\frac{\theta}{2} R) \sin\frac{\theta}{2}
    + R (1-\cos\frac{\theta}{2})
    \right)
    & \mbox{if} \quad s \geq \frac{\theta}{2} R
    \,.
  \end{cases}
\end{equation}
Obviously, it is a union of a circular arc and two semi-lines,
see Figure~\ref{Fig-1}.
Indeed, the curvature reads
\begin{equation}\label{curvature.special}
  \kappa(s) =
\begin{cases}
  0 & \mbox{if} \quad s \leq - \frac{\theta}{2} R \,,
  \\
  \frac{1}{R} & \mbox{if} \quad
  - \frac{\theta}{2} R < s < \frac{\theta}{2} R \,,
  \\
  0 & \mbox{if} \quad s \geq \frac{\theta}{2} R \,.
\end{cases}
\end{equation}
\begin{figure} [H]
\begin{center}
\includegraphics[width=.40\textwidth]{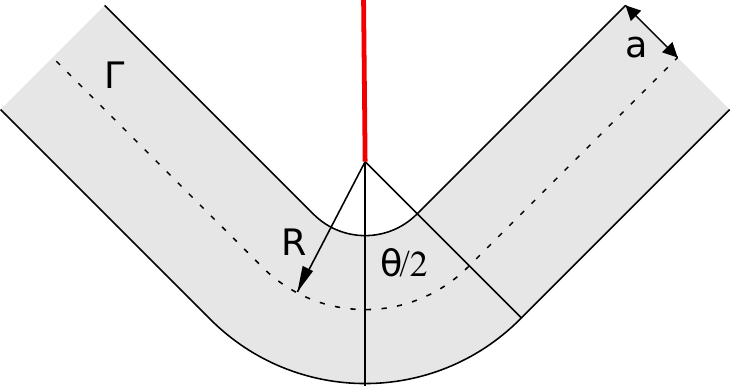}
\end{center}
\caption{The geometry of the piecewise smooth curve~\eqref{curve}
and its $a$-tubular neighbourhood}. The red line depicts
the cut-locus $\Cut(\Gamma)$.
\label{Fig-1}
\end{figure}

The case $\theta=0$ corresponds to~$\Gamma$ being a straight line,
while the other extreme situation $\theta=\pi$
is a union of a semi-circle and two parallel semi-lines,
see Figure~\ref{Fig-2}.



%
\begin{figure}[h]
\begin{center}
\begin{tabular}{ccc}
\includegraphics[width=5.5 cm]{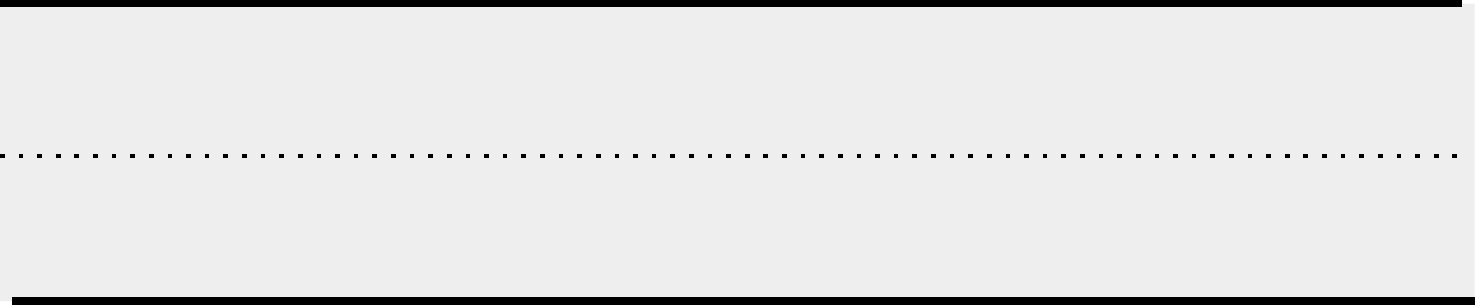}
&\rule{10ex}{0ex}&
\includegraphics[width=3.7 cm]{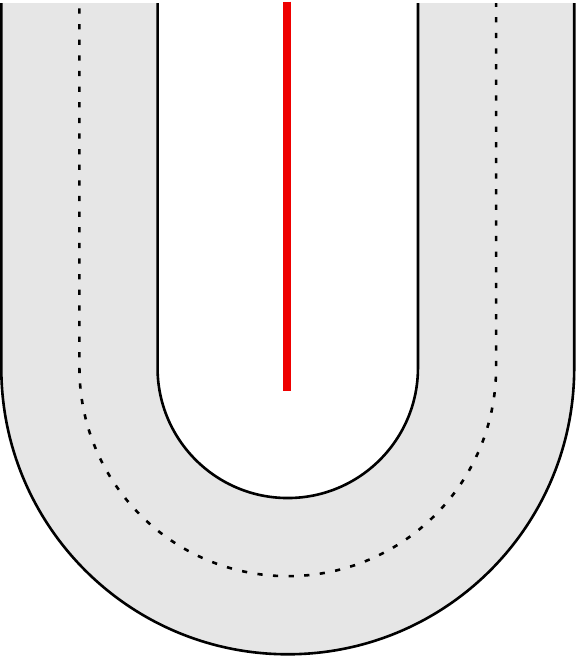}
\\
$\theta = 0$
&&
$\theta =\pi$
\end{tabular}
\caption{Extreme cases of the geometry setting~\eqref{curve}.}
\label{Fig-2}
\end{center}
\end{figure}

Except for $\theta=0$, the curve~$\Gamma$ is not $C^2$-smooth,
but it is~$C^{1,1}$-smooth and in fact piecewise analytic.
While the cut-locus of a non-smooth submanifold
in a Riemannian manifold is potentially a subtle object
(\cf~\cite{Itoh_1996}), it is reasonable in our case.
Except for $\theta=0$ when the cut-locus is empty,
it is just a semi-line
$$
  \Cut(\Gamma) = \{(0,y): \ y \geq R\}
  \qquad \mbox{whenever} \qquad
  \theta \in (0,\pi] \,.
$$
At the same time, $\Cut_0(\Gamma)=(0,R)$
for every $\theta \in (0,\pi]$.
Consequently, we have
$c_-(s) = +\infty$ for all $s \in \Real$ and
\begin{equation}\label{cut.special}
  c_+(s) =
\begin{cases}
  \frac{-s + R (\tan\frac{\theta}{2}-\frac{\theta}{2})}
  {\tan\frac{\theta}{2}}
  & \mbox{if} \quad s \leq - \frac{\theta}{2} R \,,
  \\
  R
  & \mbox{if} \quad
  - \frac{\theta}{2} R < s < \frac{\theta}{2} R \,,
  \\
  \frac{s + R (\tan\frac{\theta}{2}-\frac{\theta}{2})}
  {\tan\frac{\theta}{2}}
  & \mbox{if} \quad s \geq \frac{\theta}{2} R \,.
\end{cases}
\end{equation}
Strictly speaking, the formula for~$c_+$ makes sense
only if $\theta \in (0,\pi)$,
but the extreme situations can be recovered
after taking the respective one-sided limits
$\theta \to 0^+$ or $\theta \to \pi^-$;
namely,
$c_+(s) = +\infty$ for all $s \in \Real$ if $\theta=0$
and $c_+(s) = R$ for all $s \in \Real$ if $\theta=\pi$.

It is easy to see that~\eqref{Ass} and~\eqref{Ass.global}
hold for every $a < R$ if $\theta \in (0,\pi]$
and any~$a$ if $\theta = 0$.
The parallel coordinates described in the preceding section
extend to the present non-smooth case without any changes.
Moreover, because of the special structure of the cut-locus,
the formula~\eqref{form.transformed} becomes extremely useful
for developing the usual test-function argument.

\section{Existence of bound states}\label{Sec.proof}
%

We assume that the $V$ is either the distribution~\eqref{leaky}
or it is an essentially bounded function
satisfying~\eqref{Ass.vary}.
Then the corresponding one-dimensional operator
$$
  T := -\partial_t^2 + W(t)
  \qquad \mbox{in} \qquad
  \sii(\Real)
$$
with form domain~$H^1(\Real)$
(the sum should be understood as the form sum
in the case~\eqref{leaky}) has the essential spectrum covering $[0,\infty)$
in both cases.
From now we assume that~$W$ is attractive in the sense that
\begin{equation}\label{Ass.attractive}
  \fbox{
  \mbox{$T$ possesses at least one simple negative eigenvalue.
  }}
\end{equation}
\begin{Remark}
Hypothesis~\eqref{Ass.attractive} always holds in the leaky case~\eqref{leaky},
because~$\alpha$ is assumed to be a negative constant.
In general, a sufficient condition to guarantee~\eqref{Ass.attractive}
is that $\int W\, < 0$
(which particularly involves negative potentials of~\cite{Exner_2020}).
Moreover, it is easy to design potentials which simultaneously
satisfy $\int W\, \geq  0$ and~\eqref{Ass.attractive}
(\eg, it is enough to consider the strong coupling regime
of any~$W$ possessing a negative minimum, see~\cite[Thm.~4]{FK1}).
\end{Remark}

Let $E_1<0$ denote the lowest discrete eigenvalue of~$T$
(explicitly, $E_1 = -\frac{\alpha^2}{4}$ in the case~\eqref{leaky}).
It is well known that~$E_1$ is simple
and that the corresponding eigenfunction~$\xi_1$
can be chosen to be positive.
We additionally choose the eigenfunction
to be normalised to~$1$ in $\sii(\Real)$.
Explicitly,
$
  \xi_1(t) = \sqrt{\frac{|\alpha|}{4}} e^{\frac{\alpha}{2} |t|}
$
in the leaky case~\eqref{leaky}.
In any case, one knows that
$\xi_1 \in H^1(\Real) \cap L^\infty(\Real)$
and that the following identities hold true:
\begin{equation}\label{asymptotics}
  \xi_1(t)
  = N_\pm e^{\mp\sqrt{-E_1} t}
  \qquad \mbox{for every} \qquad
  \pm t >  a
  \,,
\end{equation}
where~$N_\pm$ are positive constants.

At the same time, for every positive number~$R$,
let us consider the double-well operator
$$
  T_R := -\partial_t^2 + W(t-R) + W(-t-R)
  \qquad \mbox{in} \qquad
  \sii(\Real)
  \,.
$$
Again, $\sigma_\mathrm{ess}(T_R) = [0,\infty)$
and~\eqref{Ass.attractive} ensures that~$T_R$
possesses at least one negative eigenvalue.
Let us denote by~$E_{1,R}$ the lowest one
and let~$\xi_{1,R}$ be a corresponding eigenfunction.
One has the following strict inequality.

%
\begin{Proposition}\label{Prop.strict}
Assume~\eqref{Ass.attractive}. Then
\begin{equation}\label{strict}
  E_{1,R} < E_1 \,.
\end{equation}
\end{Proposition}
\begin{proof}
The variational definition of~$E_{1,R}$ reads
$$
  E_{1,R} = \inf_{\stackrel[\psi \not= 0]{}{\psi \in H^1(\Real)}}
  \frac{Q_R[\psi]}{\displaystyle \int_\Real |\psi(t)|^2 \, \der t}
  \,,
$$
where
$$
  Q_R[\psi] := \int_\Real |\dot\psi(t)|^2 \, \der t
  + \int_\Real W(t+R) \, |\psi(t)|^2 \, \der t
  + \int_\Real W(-t-R) \, |\psi(t)|^2 \, \der t
  \,.
$$
Using
$$
  \psi(t) :=
  \begin{cases}
    \xi_1(t-R) & \mbox{if} \quad t \geq  0 \,,
    \\
    \xi_1(-t-R) & \mbox{if} \quad t < 0 \,,
  \end{cases}
$$
as the test function and integrating by parts, one finds
$$
  Q_R[\psi] = E_1 \int_\Real |\psi(t)|^2 \, \der t
  + \big[\bar\psi\dot\psi\big]_{0^+}^{0^-}
  \,,
$$
where
$$
  \big[\bar\psi\dot\psi\big]_{0^+}^{0^-}
  = \xi_1 (-R) \left( \dot\xi_1 (-R^ -) - \dot\xi_1 (-R^ +) \right)
  = - 2\sqrt{-E_1} \, N_-^2 e^{-2 \sqrt{-E_1} R} \,
  < 0 \,,
$$
which proves the desired claim.
\end{proof}
%
Note that the inequality  $E_{1,R} < E_1$ holds in the case~\eqref{leaky} as well. This follows
directly from the  argument derived in the proof of \cite[Lem.~2.3]{KK4}.

Now we are in a position to localise
the essential spectrum of~$H$.
\begin{Proposition}\label{Prop.ess}
Let~$\Gamma$ be given by~\eqref{curve}.
Let~$V$ satisfy~\eqref{Ass.vary} and~\eqref{Ass.attractive}.
Then
$$
  \sigma_\mathrm{ess}(H) =
\begin{cases}
  [E_1,\infty) & \mbox{if} \quad \theta \in [0,\pi) \,,
  \\
  [E_{1,R},\infty) & \mbox{if} \quad \theta=\pi \,.
\end{cases}
$$
\end{Proposition}
\begin{proof}
The case $\theta=0$ follows trivially by a separation of variables
(in fact, $\sigma(H)=[E_1,\infty)$ if $\theta=0$).
The cases $\theta \in [0,\pi)$
are due to \cite[Prop.~3.1]{Exner_2008}
(the lack of smoothness is not essential for
the arguments given there).
It remains to analyse the pathological situation $\theta=\pi$.
The result is intuitively clear because the essential spectrum
is determined by the behaviour at infinity only.
To prove it rigorously,
we proceed similarly to \cite[Sec.~3]{KK4}.

To show that $\inf\sigma_\mathrm{ess}(H) \geq E_{1,R}$,
we divide~$\Real^2$ into three subdomains
$$
\begin{aligned}
  \Omega_1 &:= \Real \times (R,\infty) \,, \\
  \Omega_2 &:= (-2R,2R) \times (-R,R) \,, \\
  \Omega_3 &:= \Real^2 \setminus \overline{(\Omega_1 \cup \Omega_2)}
  \,,
\end{aligned}
$$
and consider an auxiliary operator~$H^N$ which acts
in the same way as~$H$ but satisfies Neumann conditions
on the boundaries of the subdomains.
Since $H^N = H_1^N \oplus H_2^N \oplus H_3^N$,
where~$H_\iota^N$ with $\iota \in \{1,2,3\}$
is an operator in $\sii(\Omega_\iota)$
which acts in the same way as~$H$ but satisfies Neumann conditions
on~$\partial\Omega_\iota$, the minimax principle implies
$$
  \inf\sigma_\mathrm{ess}(H)
  \geq \min\sigma_\mathrm{ess}(H^N)
  = \min\left\{
  \inf\sigma_\mathrm{ess}(H_1^N),
  \inf\sigma_\mathrm{ess}(H_2^N),
  \inf\sigma_\mathrm{ess}(H_3^N)
  \right\}
  .
$$
Since~$H_2^N$ acts in a regular bounded domain,
its spectrum is purely discrete,
so it does not contribute
(conventionally, $\inf\sigma_\mathrm{ess}(H_2^N)=\infty$).
Since the subdomain~$\Omega_3$ does not intersect
the support of~$V$,
the operator~$H_3^N$ acts as the Laplacian, so
$
  \inf\sigma_\mathrm{ess}(H_3^N)
  \geq \inf\sigma(H_3^N)
  \geq 0
$.
Finally, the spectral problem for~$H_1^N$
can be found by a separation of variables
with the result
$
  \sigma(H_1^N)
  = \sigma_\mathrm{ess}(H_1^N)
  = [E_{1,R},\infty)
$.

To show that $\sigma_\mathrm{ess}(H) \supset [E_{1,R},\infty)$,
we construct an explicit Weyl sequence by
mollifying the function $(x,y) \mapsto \xi_1(x) \, e^{iky}$
and localising it at infinity $y \to \infty$
We refer to \cite[Sec.~3.2]{KK4} for more details.
\end{proof}
\begin{Remark}
The part of the proposition for $\theta=\pi$
shows that condition~\cite[Ass.~(c)]{Exner_2008}
is necessary to have the stability of the essential spectrum.
\end{Remark}

Now we turn to the existence of the spectrum below
the bottom of the essential spectrum.
\begin{Theorem}\label{Thm.main}
Let~$\Gamma$ be given by~\eqref{curve}.  Moreover, assume
that~$V$ satisfy~\eqref{Ass.vary} and  \eqref{Ass.attractive}.
%
%
If $\theta \in (0,\pi)$,
then
$$
  \inf\sigma(H) < E_1 \,.
$$
\end{Theorem}
\begin{proof}
Let us introduce the shifted form
$\tilde{Q}[\psi] := Q[\psi] - E_1 \|\psi\|^2$.
It is enough to find a test function $\psi \in \dom Q$
such that $\tilde{Q}[\psi] < 0$.
Then the desired result follows by the minimax principle.
The idea which comes back to~\cite{GJ}
(see~\cite{KKriz} for necessary mathematical adaptations)
is to use a mollification of~$\xi_1$ as the test function.

Given $n \in \Nat^*:=\Nat\setminus\{0\}$
(natural numbers contain zero in our convention),
let $\varphi_n:\Real \to \Real$
be the real-valued function satisfying
$\varphi_n(s)=1$ for $|s| \leq n$,
$\varphi_n(s)=0$ for $|s| \geq 2n$
and $\varphi_n(s)=(2n-|s|)/n$ for $n<|s|<2n$.
Then $\|\dot\varphi_n\|_{\sii(\Real)}^2 = 2/n \to 0$
and $\varphi_n \to 1$ pointwise as $n \to \infty$.

We define $\psi_n(s,t) := \varphi_n(s) \xi_1(t)$
and choose~$n$ so large that $\varphi_n=1$
on the interval $I := (-\frac{\theta}{2}R,\frac{\theta}{2}R)$
on which~$\kappa$ is non-trivial.
Because of the symmetry of~$\Gamma$,
one really has $\psi_n \in \dom Q$.
Let us write $\tilde{Q} = \tilde{Q}_1+\tilde{Q}_2$,
where the parts $\tilde{Q}_1$, $\tilde{Q}_2$
are defined below.

For the first part~$\tilde{Q}_1$, we have
$$
\begin{aligned}
  \tilde{Q}_1[\psi_n]
  := \ &
  \int_U \frac{|\partial_s\psi_n|^2}{1-\kappa(s) t} \, \der s \, \der t
  \\
  = \ & \ \int_\Real \int_{-c_-(s)}^{c_+(s)}
  |\dot\varphi_n(s)|^2 |\xi_1(t)|^2
  \, \der t \, \der s
  \\
  \leq \ & \|\dot\varphi_n\|_{\sii(\Real)}^2 \|\xi_1(t)\|_{\sii(\Real)}^2
  \\
  = \ & \frac{2}{n} \xrightarrow[n \to \infty]{} 0
  \,.
\end{aligned}
$$

For the second part~$\tilde{Q}_2$,
integrating by parts (twice),
we have
\begin{equation*}
\begin{aligned}
  \tilde{Q}_2[\psi_n]
  := \ &
  \int_U |\partial_t\psi_n|^2 (1-\kappa(s) t) \, \der s \, \der t
  + \int_U W(t) \, |\psi_n|^2 (1-\kappa(s) t) \, \der s \, \der t
  - E_1 \int_U |\psi_n|^2 (1-\kappa(s) t) \, \der s \, \der t
  \\
  = \ & \int_U \kappa(s) |\varphi_n(s)|^2 \xi_1(t) \dot\xi_1(t)
  \, \der s \, \der t
  + \int_\Real |\varphi_n(s)|^2
  \left[ \xi_1(t) \dot{\xi}_1(t) (1-\kappa(s)t)
  \right]_{t=-c_-(s)}^{t=c_+(s)} \ \der s
  \\
  = \ & \int_\Real |\varphi_n(s)|^2
  \left[
  \frac{1}{2} \xi_1(t)^2 \kappa(s)
  +\xi_1(t) \dot{\xi}_1(t) (1-\kappa(s)t)
  \right]_{t=-c_-(s)}^{t=c_+(s)} \der s
  \,.
\end{aligned}
\end{equation*}
Hence,
$
  \tilde{Q}_2[\psi_n] = \tilde{Q}_2^\mathrm{int}[\psi_n]
  +\tilde{Q}_2^\mathrm{ext}[\psi_n]
$,
where the forms on the right-hand side correspond
to dividing the last integral to an integration
over~$I$ and $\Real\setminus I$, respectively.
Explicitly,
\begin{align*}
  \tilde{Q}_2^\mathrm{int}[\psi_n]
  &= \int_I
  \left[
  \frac{1}{2} \xi_1(t)^2 \kappa(s)
  +\xi_1(t) \dot{\xi}_1(t) (1-\kappa(s)t)
  \right]_{t=-c_-(s)}^{t=c_+(s)} \ \der s
  \,,
  \\
  \tilde{Q}_2^\mathrm{ext}[\psi_n]
  &= \int_{\Real\setminus I} |\varphi_n(s)|^2
  \left[
  \xi_1(t) \dot{\xi}_1(t)
  \right]_{t=-c_-(s)}^{t=c_+(s)} \der s
  =  - \sqrt{-E_1} \int_{\Real\setminus I} |\varphi_n(s)|^2
  \left[
  \xi_1(c_+(s))^2 + \xi_1(-c_-(s))^2
  \right] \der s
  \,,
\end{align*}
where the last equality is due to~\eqref{asymptotics}.
Here $\tilde{Q}_2^\mathrm{int}[\psi_n]$ is of course finite
and in fact independent of~$n$ (because $\varphi_n=1$ on~$I$).
Since $\tilde{Q}_2^\mathrm{ext}[\psi_n]$ is negative for every~$n$,
the limit
$$
  \lim_{n \to \infty} \tilde{Q}_2^\mathrm{ext}[\psi_n]
  = - \sqrt{-E_1} \int_{\Real\setminus I}
  \left[
  \xi_1(c_+(s))^2 + \xi_1(-c_-(s))^2
  \right] \der s
$$
is well defined
(it can be $-\infty$, namely if $\theta=-\pi$,
see Remark~\ref{Rem.pi} below).

In summary,
\begin{multline*}
  \lim_{n \to \infty} \tilde{Q}[\psi_n]
  = \tilde{Q}_2^\mathrm{int}[\psi_n]
  + \lim_{n \to \infty} \tilde{Q}_2^\mathrm{ext}[\psi_n]
  \\
  =
  \int_I
  \left[
  \frac{1}{2} \xi_1(t)^2 \kappa(s)
  +\xi_1(t) \dot{\xi}_1(t) (1-\kappa(s)t)
  \right]_{t=-c_-(s)}^{t=c_+(s)} \ \der s
  -  \sqrt{-E_1} \int_{\Real\setminus I}
  \left[
  \xi_1(c_+(s))^2 + \xi_1(-c_-(s))^2
  \right] \der s
  \,.
\end{multline*}
Using the special form of~$\Gamma$ in this general result,
namely the formulae~\eqref{curvature.special} and~\eqref{cut.special},
one finds
\begin{equation*}
  \tilde{Q}_2^\mathrm{int}[\psi_n]
  = \xi_1(R)^2 \mbox{$\frac{\theta}{2}$} 
  \,, \qquad
  \lim_{n \to \infty} \tilde{Q}_2^\mathrm{ext}[\psi_n]
  = - \xi_1(R)^2\tan\mbox{$\frac{\theta}{2}$}
  \,.
\end{equation*}
We finally arrive at
\begin{equation}\label{result}
  \lim_{n \to \infty} \tilde{Q}[\psi_n]
  = \xi_1(R)^2 \,
  (\mbox{$\frac{\theta}{2}$} -\tan\mbox{$\frac{\theta}{2}$})<0
  \,.
\end{equation}
%
In these circumstances, we can therefore find a positive~$n_0$
such that $\tilde{Q}[\psi_n] < 0$ for every $n \geq n_0$.
\end{proof}
\begin{Remark}\label{Rem.pi}
The theorem is valid also for $\theta=\pi$.
Indeed, then the right-hand side equals~$-\infty$,
because
$
  \tilde{Q}_2^\mathrm{ext}[\psi_n]
$
tends to~$-\infty$ as $n \to \infty$
while $\tilde{Q}_2^\mathrm{int}[\psi_n]$
is independent of~$n$.
Hence the conclusion of the proof holds as well.
However, in this case, the result also follows from
Propositions~\ref{Prop.strict} and~\ref{Prop.ess}.
\end{Remark}

As a direct combination of Theorem~\ref{Thm.main}
with Proposition~\ref{Prop.ess},
we get the following ultimate result
about the existence of discrete spectra
in soft and leaky waveguides.
\begin{Corollary}
Let~$\Gamma$ be given by~\eqref{curve}.  Moreover, assume that~$V$ satisfy~\eqref{Ass.vary} and \eqref{Ass.attractive}.
If $\theta \in (0,\pi)$,
then
$$
  \sigma_\mathrm{disc}(H) \not= \varnothing \,.
$$
\end{Corollary}
\begin{Remark} Note that the discrete spectrum is obviously empty for
the straight waveguide corresponding to $\theta=0$.
In this paper we leave the problem of the existence of discrete eigenvalues for
$\theta =\pi $ open. Also it arises a natural question about spectral properties of the system if $\theta \to \pi^-$. One may expect that in this case the number of discrete eigenvalues goes  to infinity. Therefore the question is: what is the asymptotics of the counting function for $\theta \to \pi^-$?
\end{Remark}

\subsection*{Acknowledgment}
D.K.\ was supported
by the EXPRO grant No.~20-17749X
of the Czech Science Foundation.
J. K. was supported by Internal Project of Excellent Research of the Faculty of Science of University Hradec Kr\'alov\'e, ``Studying of properties of confined quantum particle using Woods-Saxon potential''.

\providecommand{\bysame}{\leavevmode\hbox
to3em{\hrulefill}\thinspace}
\providecommand{\MR}{\relax\ifhmode\unskip\space\fi MR }
\providecommand{\MRhref}[2]{%
  \href{http://www.ams.org/mathscinet-getitem?mr=#1}{#2}
} \providecommand{\href}[2]{#2}

\end{document}